\documentclass[a4paper,11pt]{amsart}
\usepackage{enumitem}
\usepackage[hidelinks]{hyperref}

\def\reusesizes{
\usepackage[a4paper]{geometry}
\setlength{\textheight}{19.8cm}
\setlength{\textwidth}{13.5cm}
\setlength{\oddsidemargin}{1.2cm}
\setlength{\evensidemargin}{1.2cm}
\setlength{\topmargin}{1.65cm} 
}
\reusesizes

\usepackage[small]{caption}
\usepackage[subrefformat=parens]{subcaption}
\usepackage{tikz}
\usetikzlibrary{calc}

\title{Graph classes equivalent to 12-representable graphs}
\author[A.~Takaoka]{Asahi~Takaoka}
\address{
  College of Information and Systems, 
  Muroran Institute of Technology, 
  Mizumoto 27-1, Muroran, 
  Hokkaido, 050--8585, Japan 
}
\email{takaoka@mmm.muroran-it.ac.jp}
\date{\today}
\subjclass[2010]{05C62, 05C75}
\keywords{12-representable graphs, forbidden induced subgraphs, interval containment bigraphs, simple-triangle graphs, vertex ordering characterization}

\newtheorem{theorem}{Theorem}

\newtheorem{cor}[theorem]{Corollary}
\theoremstyle{definition}
\newtheorem{dnt}[theorem]{Definition}
\newtheorem{exm}[theorem]{Example}
\newtheorem{rem}[theorem]{Remark}

\def\red{\textup{red}}

\begin{document}
\maketitle
\begin{abstract}
Jones et al. (2015) introduced the notion of $u$-representable graphs, 
where $u$ is a word over $\{1, 2\}$ different from $22\cdots2$, 
as a generalization of word-representable graphs. 
Kitaev (2016) showed that if $u$ is of length at least 3, 
then every graph is $u$-representable. 
This indicates that there are only two nontrivial 
classes in the theory of $u$-representable graphs: 
11-representable graphs, which correspond to word-representable graphs, 
and 12-representable graphs. 
This study deals with 12-representable graphs. 
\par
Jones et al. (2015) provided 
a characterization of 12-representable trees in terms of forbidden induced subgraphs. 
Chen and Kitaev (2022) presented 
a forbidden induced subgraph characterization of a subclass of 12-representable grid graphs. 
\par
This paper shows that a bipartite graph is 12-representable 
if and only if it is an interval containment bigraph. 
The equivalence gives us a forbidden induced subgraph characterization 
of 12-representable bipartite graphs 
since the list of minimal forbidden induced subgraphs is known 
for interval containment bigraphs. 
We then have a forbidden induced subgraph characterization for grid graphs, which solves an open problem of Chen and Kitaev (2022). 
The study also shows that a graph is 12-representable 
if and only if it is the complement of a simple-triangle graph. 
This equivalence indicates that 
a necessary condition for 12-representability presented by Jones et al. (2015) is also sufficient. 
Finally, we show from these equivalences that 12-representability can be determined 
in $O(n^2)$ time for bipartite graphs and 
in $O(n(\bar{m}+n))$ time for arbitrary graphs, 
where $n$ and $\bar{m}$ are the number of vertices and edges of the complement of the given graph. 

\end{abstract}

\section{Introduction}
The notion of $u$-representable graphs, 
where $u$ is a word over $\{1, 2\}$ different from $22\cdots2$, 
was introduced by Jones et al.~\cite{JKPR15-EJC} 
as a generalization of a well-studied class of 
word-representable graphs~\cite{KL15-book,Kitaev17-LNCS}. 
In this context, word-representable graphs 
correspond to 11-representable graphs. 
\par
Jones et al.~\cite{JKPR15-EJC} showed that 
any graph is $1^k$-representable for every $k \geq 3$, 
where $1^k$ denotes $k$ concatenated copies of 1. 
Extending this result, Kitaev~\cite{Kitaev17-JGT} showed that 
for every $u \in \{1, 2\}^*$ of length at least 3, any graph is $u$-representable. 
Therefore, only two graph classes are nontrivial 
in the theory of $u$-representable graphs: 
11-representable graphs and 12-representable graphs. 
This paper focuses on 12-representable graphs. 
Note that the class of 21-representable graphs is equivalent to 
that of 12-representable graphs, as shown in the next section. 
\par
The class of 12-representable graphs is 
a proper subclass of comparability graphs and 
a proper superclass of 
both co-interval graphs and permutation graphs~\cite{JKPR15-EJC}. 
The class of 12-representable graphs is not equivalent to 
that of 11-representable graphs 
since word-representable graphs (i.e., 11-representable graphs) 
generalize comparability graphs~\cite{KL15-book,Kitaev17-LNCS}. 
It is also known that 
any cycle of length at least 5 is not 12-representable~\cite{JKPR15-EJC}. 
This implies that 12-representable graphs are weakly chordal 
since the graph $\overline{C_n}$ (the complement of the cycle of length $n$) 
is not a comparability graph for any $n \geq 5$, see, e.g.,~\cite{Gallai67} and~\cite[Corollary 2.11]{GT04}. 
\par
Jones et al.~\cite{JKPR15-EJC} showed that 
a tree is 12-representable if and only if it is a \emph{double caterpillar}, 
a tree in which every vertex is within distance 2 from a central path. 
It is easy to see that a tree is a double caterpillar if and only if 
it contains no $T_3$ as a subtree (see, e.g.,~\cite[Lemma~18]{STU10-DAM}), 
where $T_3$ is the tree in Figure~\ref{fig:Forbidden subgraphs}\subref{fig:T3}. 
They also initiated the study of the 12-representability of grid graphs. 
They provided some 12-representable grid graphs and 
asked whether such graphs could be characterized.
We use the term \emph{grid graph} in this paper to mean 
an induced subgraph of a rectangular grid graph. 
\par
Chen and Kitaev~\cite{CK22-DMGT} answered this question. 
They called a grid graph a \emph{square grid graph} 
if every edge belongs to a cycle of length 4, 
and showed that a square grid graph is 12-representable if and only if 
it contains no $X$ and no cycle of length $2n$ for $n \geq 4$ 
as an induced subgraph, 
where $X$ is the graph in Figure~\ref{fig:Forbidden subgraphs}\subref{fig:X}. 
They also provided a conjecture for characterizing 12-representable \emph{line grid graph}, 
grid graphs that are not square grid graphs~\cite[Conjecture 3.6]{CK22-DMGT}. 
We will deal with this conjecture in Remark~\ref{rem:Forbidden subgraphs}. 
\par
Meanwhile, Jones et al.~\cite{JKPR15-EJC} gave 
a necessary condition for the 12-representability of a graph 
in terms of graph labelings (see Theorem~\ref{thm:good labeling}). 
Chen and Kitaev~\cite{CK22-DMGT} showed that the necessary condition 
is also sufficient for square grid graphs. 
Whether the condition is sufficient for arbitrary graphs 
was left as an open question~\cite{CK22-DMGT}. 
\par
This study shows that a bipartite graph is 12-representable 
if and only if it is an interval containment bigraph~\cite{Huang18-DAM}. 
We also demonstrate that a graph is 12-representable 
if and only if it is the complement of a simple-triangle graph~\cite{CK87-CN}. 
These equivalences provide some structural results on 12-representable graphs. 
In particular, we obtain a forbidden induced subgraph characterization 
of 12-representable bipartite graphs and then also for grid graphs. 
Moreover, we obtain from a characterization of simple-triangle graphs~\cite{Takaoka18-DM} 
that the necessary condition of Jones et al.~\cite{JKPR15-EJC} mentioned above is in fact also sufficient. 

\section{Preliminaries}
This section presents some definitions, notations, and results used in this paper. 

All graphs in this paper are finite, simple, and undirected. 
We write $xy$ for the edge joining two vertices $x$ and $y$. 
For a graph $G$, we write $V(G)$ and $E(G)$ 
for the vertex set and the edge set of $G$, respectively. 
We usually denote the number of vertices and edges by $n$ and $m$, respectively. 
The \emph{complement} of a graph $G$ is the graph $\overline{G}$ 
such that $V(\overline{G}) = V(G)$ and 
$xy \in E(\overline{G})$ if and only if $xy \notin E(G)$ 
for any two distinct vertices $x, y \in V(\overline{G})$.

\subsection{Words and 12-representable graphs}
For a positive integer $n$, 
let $[n] = \{1, 2, \ldots, n\}$ and $[n]^*$ be the set of all words over $[n]$. 
For a word $w \in [n]^*$, 
let $A(w)$ denote the set of letters occurring in $w$. 
For a subset $B \subseteq A(w)$, 
let $w_B$ be a word obtained from $w$ by removing all the letters of $A(w) \setminus B$. 
For a word $w \in [n]^*$, the \emph{reduced form} of $w$, denoted by $\red(w)$, 
is the word obtained from $w$ by replacing 
each occurrence of the $i$th smallest letter with $i$. 
Let $u = u_1 u_2 \cdots u_k$ with $\red(u) = u$. 
A word $w = w_1 w_2 \cdots w_m$ of $[n]^*$ has a \emph{$u$-match} 
if there is an index $i$ such that $\red(w_i w_{i+1} \cdots w_{i+k-1}) = u$, 
that is, up to reduction, $u$ occurs consecutively in $w$. 
\par
A \emph{labeled graph} of a graph $G$ is obtained from $G$ 
by assigning an integer (label) to each vertex. 
This paper assumes that all labels are distinct and from $[n]$, 
where $n$ denotes the number of vertices of the graph. 
Given a word $u \in [2]^*$ such that $\red(u) = u$ 
(i.e., $u$ is different from $22\cdots2$), 
a labeled graph $G$ is \emph{$u$-representable} if 
there is a word $w \in [n]^*$ 
such that $A(w) = [n]$ and for any $x, y \in V(G)$, 
$xy \in E(G)$ if and only if $w_{\{x, y\}}$ has no $u$-matches. 
In this case, we say that the word $w$ \emph{$u$-represents} the graph $G$ 
and $w$ is a \emph{$u$-representant} of $G$. 
An unlabeled graph $H$ is \emph{$u$-representable} if 
there is a labeling of $H$ such that the resulting labeled graph $H'$ is $u$-representable. 
\par
By definition, the class of $u$-representable graphs is hereditary 
(i.e., closed under taking induced subgraphs). 
Note also that the class of $u$-representable graphs is equivalent 
to that of $u^r$-representable graphs, 
where $u^r$ denotes the reverse of $u$, 
since if a word $w$ is a $u$-representant of a graph $G$, 
then its reverse $w^r$ is a $u^r$-representant of $G$ and vice versa. 
Thus, as noted in the introduction, 
the classes of 12-representable and 21-representable graphs 
are equivalent.

\subsection{Necessary condition}
Given a labeled graph $G$, 
the \emph{reduced form} of $G$, denoted by $\red(G)$, 
is the labeled graph obtained from $G$ by relabeling 
so that the $i$th smallest label is replaced by $i$. 
For a graph $G$, a graph $H$ is an \emph{induced subgraph} 
if $V(H) \subseteq V(G)$ and $xy \in E(H) \iff xy \in E(G)$ 
for any $x, y \in V(H)$. 
We will use the following necessary condition to determine 12-representable graphs. 
\begin{theorem}[\cite{JKPR15-EJC}]\label{thm:good labeling}
Let $G$ be a labeled graph. 
If $G$ has an induced subgraph $H$ such that 
$\red(H)$ is equal to one of $I_3$, $J_4$, or $Q_4$ in Figure~\ref{fig:I3 J4 Q4}, 
then $G$ is not $12$-representable. 
\end{theorem}

\begin{figure}[ht]
  \centering
  \subcaptionbox{\label{fig:I3}}{\begin{tikzpicture}
\useasboundingbox (-1.5, -.7) rectangle (1.5, 0.7);
\tikzstyle{every node}=[draw,circle,fill=white,minimum size=5pt,inner sep=0pt]
\node [label=above:$1$] (x) at (-1, 0) {};
\node [label=above:$2$] (y) at ( 0, 0) {};
\node [label=above:$3$] (z) at ( 1, 0) {};
\draw [] (x) -- (y) -- (z);
\end{tikzpicture}}
  \subcaptionbox{\label{fig:J4}}{\begin{tikzpicture}
\useasboundingbox (-1.5, -.7) rectangle (1.5, 0.7);
\tikzstyle{every node}=[draw,circle,fill=white,minimum size=5pt,inner sep=0pt]
\node [label= left:$1$] (x) at (-.5,  .5) {};
\node [label=right:$3$] (y) at ( .5,  .5) {};
\node [label= left:$2$] (z) at (-.5, -.5) {};
\node [label=right:$4$] (w) at ( .5, -.5) {};
\draw [] (x) -- (y) (z) -- (w);
\end{tikzpicture}}
  \subcaptionbox{\label{fig:Q4}}{\begin{tikzpicture}
\useasboundingbox (-1.5, -.7) rectangle (1.5, 0.7);
\tikzstyle{every node}=[draw,circle,fill=white,minimum size=5pt,inner sep=0pt]
\node [label= left:$1$] (x) at (-.5,  .5) {};
\node [label=right:$4$] (y) at ( .5,  .5) {};
\node [label= left:$2$] (z) at (-.5, -.5) {};
\node [label=right:$3$] (w) at ( .5, -.5) {};
\draw [] (x) -- (y) (z) -- (w);
\end{tikzpicture}}
  \caption{Forbidden labeled graphs $I_3$~\subref{fig:I3}, $J_4$~\subref{fig:J4}, and $Q_4$~\subref{fig:Q4}. }
  \label{fig:I3 J4 Q4}
\end{figure}
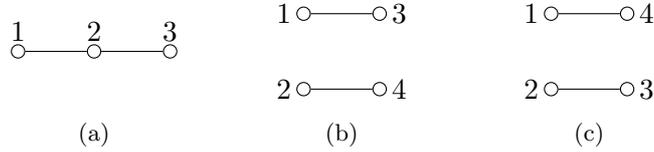

We now define the notion of $F$-free labeling. 
\begin{dnt}\label{def:F-free}
Let $F$ be a set of labeled graphs. 
(Recall that we assume all labels are distinct and from $[n]$, where $n$ is the number of vertices of the graph.) 
A graph labeling is \emph{$F$-free} 
if it contains no induced subgraphs in $F$ in the reduced form. 
\end{dnt}
Note that $\{I_3, J_4, Q_4\}$-free labeling is said to be \emph{good} 
by Chen and Kitaev~\cite{CK22-DMGT}. 
They showed that the existence of a good labeling for a square grid graph 
implies that the graph is 12-representable.

\subsection{Interval containment bigraphs}
A graph $G$ is \emph{bipartite} if $V(G)$ can be partitioned into 
two independent sets $X$ and $Y$. 
Such a partition $(X, Y)$ is called a \emph{bipartition} of $G$. 
A bipartite graph $G$ with bipartition $(X, Y)$ is 
an \emph{interval containment bigraph}~\cite{Huang18-DAM} 
if there is an interval $I_v$ for each vertex $v \in V(G)$ 
such that for any $x \in X$ and $y \in Y$, 
$xy \in E(G)$ if and only if $I_x$ contains $I_y$. 
The set $\{I_v \colon\ v \in V(G)\}$ is called a 
\emph{model} or \emph{representation} of $G$. 
See Figures~\ref{fig:ICB}\subref{fig:ICB-graph} and~\ref{fig:ICB}\subref{fig:ICB-model} for example. 

\begin{figure}[ht]
  \centering
  \subcaptionbox{\label{fig:ICB-graph}}{\begin{tikzpicture}
\def\len{0.8}
\useasboundingbox (-2.5*\len, -1.5*\len) rectangle (2.5*\len, 2*\len);
\tikzstyle{every node}=[draw,circle,fill=white,minimum size=5pt,inner sep=0pt]
\node [label=above:4]        (b) at ($(45:\len*1) + (135:\len*1)$) {};
\node [label=above right:8, fill=black]  (c) at ($(45:\len*1) + (135:\len*0)$) {};
\node [label=right:6]        (d) at ($(45:\len*1) + (135:\len*-1)$) {};
\node [label=below:7, fill=black]        (e) at ($(45:\len*0) + (135:\len*-1)$) {};
\node [label=above:1]        (f) at (0, 0) {};
\node [label=above left:5, fill=black]   (g) at ($(45:\len*0) + (135:\len*1)$) {};
\node [label=left:2]         (h) at ($(45:\len*-1) + (135:\len*1)$) {};
\node [label=below:3, fill=black]        (i) at ($(45:\len*-1) + (135:\len*0)$) {};
\draw [] 
	(b) -- (c) -- (d)
	(e) -- (f) -- (g)
	(h) -- (i)
	(b) -- (g) -- (h)
	(c) -- (f) -- (i)
	(d) -- (e)
;
\end{tikzpicture}}
  \subcaptionbox{\label{fig:ICB-model}}{\begin{tikzpicture}
\def\xlen{0.45}
\def\ylen{0.4}
\useasboundingbox (-1*\xlen, -3*\ylen) rectangle (18.5*\xlen, 4*\ylen);
\tikzstyle{every node}=[inner sep=0.5pt]
\node [label=left:$X$] at (0.5*\xlen, 3/2*\ylen) {};
\node [label=left:$Y$] at (0.5*\xlen, -3/2*\ylen) {};
\draw [] (0, 0) -- (18*\xlen, 0);
\node [label=left:1] (1) at (1*\xlen, 3*\ylen) {};
\node [label=right:1] (1') at ($(1) + (16*\xlen, 0)$) {};
\draw [thick,{|-|}] (1) -- (1');
\node [label=left:2] (2) at (2*\xlen, 2*\ylen) {};
\node [label=right:2] (2') at ($(2) + (6*\xlen, 0)$) {};
\draw [thick,{|-|}] (2) -- (2');
\node [label=left:3] (3) at (3*\xlen, -1*\ylen) {};
\node [label=right:3] (3') at ($(3) + (4*\xlen, 0)$) {};
\draw [thick,{|-|}] (3) -- (3');
\node [label=left:4] (4) at (4*\xlen, 1*\ylen) {};
\node [label=right:4] (4') at ($(4) + (10*\xlen, 0)$) {};
\draw [thick,{|-|}] (4) -- (4');
\node [label=left:5] (5) at (5*\xlen, -2*\ylen) {};
\node [label=right:5] (5') at ($(5) + (1*\xlen, 0)$) {};
\draw [thick,{|-|}] (5) -- (5');
\node [label=left:6] (6) at (10*\xlen, 2*\ylen) {};
\node [label=right:6] (6') at ($(6) + (6*\xlen, 0)$) {};
\draw [thick,{|-|}] (6) -- (6');
\node [label=left:7] (7) at (11*\xlen, -1*\ylen) {};
\node [label=right:7] (7') at ($(7) + (4*\xlen, 0)$) {};
\draw [thick,{|-|}] (7) -- (7');
\node [label=left:8] (8) at (12*\xlen, -2*\ylen) {};
\node [label=right:8] (8') at ($(8) + (1*\xlen, 0)$) {};
\draw [thick,{|-|}] (8) -- (8');
\end{tikzpicture}}
  \subcaptionbox{\label{fig:ICB-ordering}}{\begin{tikzpicture}
\def\len{0.8}
\useasboundingbox (-0.5, -0.5) rectangle (7*\len + 0.5, 1.3);
\tikzstyle{every node}=[draw,circle,fill=white,minimum size=5pt,inner sep=0pt]
\node [label=below:$1$            ] (v1) at (0*\len, 0) {};
\node [label=below:$2$            ] (v2) at (1*\len, 0) {};
\node [label=below:$3$, fill=black] (v3) at (2*\len, 0) {};
\node [label=below:$4$            ] (v4) at (3*\len, 0) {};
\node [label=below:$5$, fill=black] (v5) at (4*\len, 0) {};
\node [label=below:$6$            ] (v6) at (5*\len, 0) {};
\node [label=below:$7$, fill=black] (v7) at (6*\len, 0) {};
\node [label=below:$8$, fill=black] (v8) at (7*\len, 0) {};
\draw [] (v1) to [out=45, in=135] (v3);
\draw [] (v1) to [out=45, in=135] (v5);
\draw [] (v1) to [out=45, in=135] (v7);
\draw [] (v1) to [out=45, in=135] (v8);
\draw [] (v2) to (v3);
\draw [] (v2) to [out=45, in=135] (v5);
\draw [] (v4) to (v5);
\draw [] (v4) to [out=45, in=135] (v8);
\draw [] (v6) to (v7);
\draw [] (v6) to [out=45, in=135] (v8);
\end{tikzpicture}}
  \caption{
    \subref{fig:ICB-graph} An interval containment bigraph $G_1$. 
    \subref{fig:ICB-model} A model of $G_1$. 
    \subref{fig:ICB-ordering} An ordering of the vertices of $G_1$. 
    White and black vertices are in $X$ and $Y$, respectively. 
    The vertices are labeled based on the left endpoints of the intervals. 
    As shown in Example~\ref{ex:ICB}, 
    the word $w = 3578.53284761.1246$ is a 12-representant of $G_1$ 
    (the dots are not part of the word, 
     they are only included as delimiters of the word parts 
     as constructed in the proof of Theorem~\ref{thm:ICB}). 
    }
  \label{fig:ICB}
\end{figure}
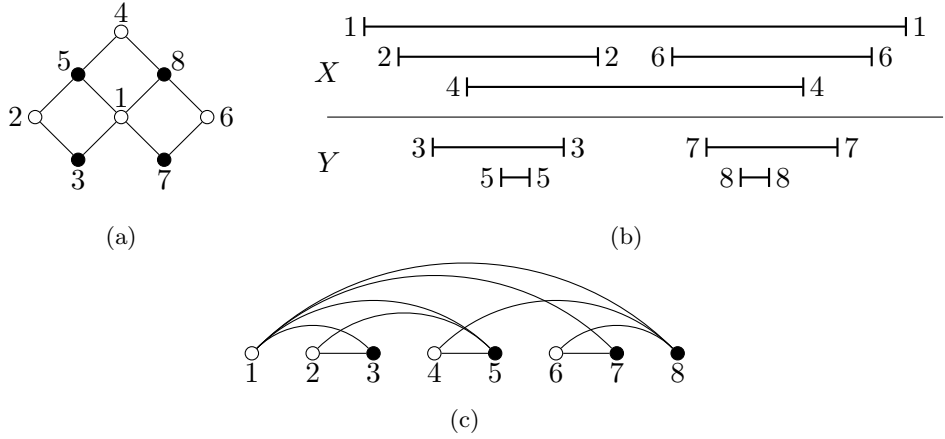

The class of interval containment bigraphs is equivalent to 
some classes of graphs, such as 
bipartite graphs whose complements are circular-arc graphs~\cite{FHH99-Combinatorica} and 
two-directional orthogonal ray graphs~\cite{STU10-DAM}. 
We will use this equivalence in Section~\ref{sec:ICB}. 
The other equivalent classes can be found in~\cite{SBSW14-DAM,TTU14-IEICE}. 
Among those, we choose the model of interval containment bigraphs 
because of the simplicity of the construction of 12-representants. 
\par
Many results have been obtained for these classes, including 
a forbidden induced subgraph characterization~\cite{TM76-DM,FHH99-Combinatorica,STU10-DAM} and 
polynomial-time recognition algorithms~\cite{STU10-DAM}. 
The class of interval containment bigraphs is a proper subclass of chordal bipartite graphs 
and a superclass of bipartite permutation graphs~\cite{STU10-DAM}.

\subsection{Simple-triangle graphs}
Let $L_1$ and $L_2$ be two horizontal lines in the plane with $L_1$ above $L_2$. 
A point on $L_1$ and an interval on $L_2$ define a triangle between $L_1$ and $L_2$. 
A graph is a \emph{simple-triangle graph} 
if there is a triangle $T_v$ for each vertex $v \in V(G)$ 
such that for any $x, y \in V(G)$, 
$xy \in E(G)$ if and only if $T_x$ intersects $T_y$. 
The set $\{T_v \colon\ v \in V(G)\}$ is called a 
\emph{model} or \emph{representation} of $G$. 
See Figures~\ref{fig:PI}\subref{fig:PI-graph} and~\ref{fig:PI}\subref{fig:PI-model} for example. 

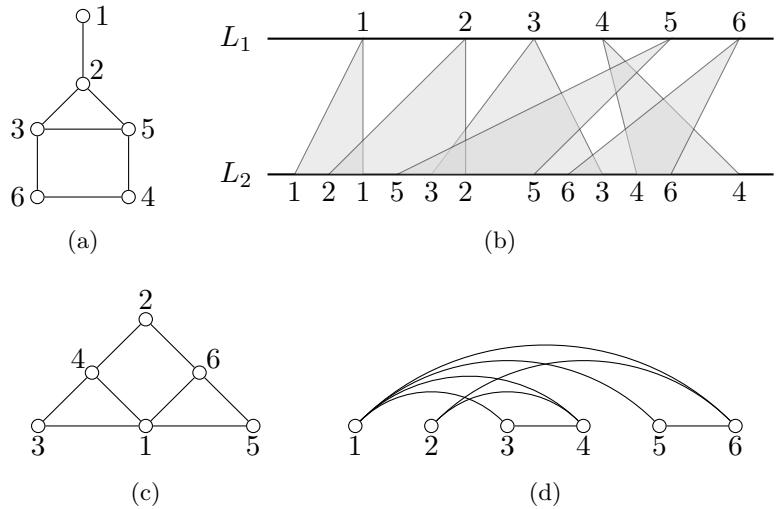
\begin{figure}[ht]
  \centering
  \subcaptionbox{\label{fig:PI-graph}}{\begin{tikzpicture}
\def\len{0.6}
\useasboundingbox (-1.4, -0.5) rectangle (1.4, 3*\len+0.5);
\tikzstyle{every node}=[draw,circle,fill=white,minimum size=5pt,inner sep=0pt]
\node [label=right:1]  (c) at (0, 3.5*\len) {};
\node [label=above right:2]  (a) at (0, 2*\len) {};
\node [label=left:3]   (b) at (-1*\len, \len) {};
\node [label=right:4]  (d) at ( 1*\len, -0.5*\len) {};
\node [label=right:5]  (e) at ( 1*\len, \len) {};
\node [label=left:6]   (f) at (-1*\len, -0.5*\len) {};
\draw [] 
	(c) -- (a) -- (e) -- (d) -- (f) -- (b) -- (a)
	(e) -- (b)
;
\end{tikzpicture}}
  \subcaptionbox{\label{fig:PI-model}}{\begin{tikzpicture}
\def\len{0.45}
\useasboundingbox (-1.5*\len, -0.5) rectangle (15.5*\len, 4*\len+0.5);
\tikzstyle{every node}=[minimum size=5pt, inner sep=0pt]
\def\La{4*\len}
\def\Lb{0*\len}
\node [label=left:$L_1$] (L1) at (0, \La) {};
\node [label=left:$L_2$] (L2) at (0, \Lb) {};
\draw [thick] (L1) -- (15*\len, \La);
\draw [thick] (L2) -- (15*\len, \Lb);
\def\pa{(3*\len, \La)}
\def\pc{(6*\len, \La)}
\def\pb{(8*\len, \La)}
\def\pd{(10*\len, \La)}
\def\pe{(12*\len, \La)}
\def\pf{(14*\len, \La)}
\node [label=above:1] at \pa {};
\node [label=above:2] at \pc {};
\node [label=above:3] at \pb {};
\node [label=above:4] at \pd {};
\node [label=above:5] at \pe {};
\node [label=above:6] at \pf {};
\def\la{(1*\len, \Lb)}
\def\lc{(2*\len, \Lb)}
\def\ra{(3*\len, \Lb)}
\def\le{(4*\len, \Lb)}
\def\lb{(5*\len, \Lb)}
\def\rc{(6*\len, \Lb)}
\def\re{(8*\len, \Lb)}
\def\lf{(9*\len, \Lb)}
\def\rb{(10*\len, \Lb)}
\def\ld{(11*\len, \Lb)}
\def\rf{(12*\len, \Lb)}
\def\rd{(14*\len, \Lb)}
\node [label=below:1] at \la {};
\node [label=below:2] at \lc {};
\node [label=below:1] at \ra {};
\node [label=below:5] at \le {};
\node [label=below:3] at \lb {};
\node [label=below:2] at \rc {};
\node [label=below:5] at \re {};
\node [label=below:6] at \lf {};
\node [label=below:3] at \rb {};
\node [label=below:4] at \ld {};
\node [label=below:6] at \rf {};
\node [label=below:4] at \rd {};
\def\op{0.5}
\draw [fill=gray!30, opacity=\op] \pa -- \la -- \ra --  cycle;
\draw [fill=gray!30, opacity=\op] \pb -- \lb -- \rb --  cycle;
\draw [fill=gray!30, opacity=\op] \pc -- \lc -- \rc --  cycle;
\draw [fill=gray!30, opacity=\op] \pd -- \ld -- \rd --  cycle;
\draw [fill=gray!30, opacity=\op] \pe -- \le -- \re --  cycle;
\draw [fill=gray!30, opacity=\op] \pf -- \lf -- \rf --  cycle;
\end{tikzpicture}}
  \subcaptionbox{\label{fig:co-PI-graph}}{\begin{tikzpicture}
\def\len{1.0}
\useasboundingbox (-2, -0.5) rectangle (2, 2.2);
\tikzstyle{every node}=[draw,circle,fill=white,minimum size=5pt,inner sep=0pt]
\node [label=above:2]        (a) at ($(45:\len*1) + (135:\len*1)$) {};
\node [label=above right:6]  (f) at ($(45:\len*1) + (135:\len*0)$) {};
\node [label=below:5]        (e) at ($(45:\len*1) + (135:\len*-1)$) {};
\node [label=below:1]        (c) at (0, 0) {};
\node [label=above left:4]   (d) at ($(45:\len*0) + (135:\len*1)$) {};
\node [label=below:3]         (b) at ($(45:\len*-1) + (135:\len*1)$) {};
\draw [] 
	(a) -- (f) -- (e) -- (c) -- (b) -- (d) -- (a)
	(d) -- (c) -- (f)
;
\end{tikzpicture}}
  \subcaptionbox{\label{fig:co-PI-ordering}}{\begin{tikzpicture}
\def\len{1.0}
\useasboundingbox (-0.5, -0.5) rectangle (5*\len + 0.5, 2.2);
\tikzstyle{every node}=[draw,circle,fill=white,minimum size=5pt,inner sep=0pt]
                        
\node [label=below:$1$] (v1) at (0*\len, 0) {};
\node [label=below:$2$] (v2) at (1*\len, 0) {};
\node [label=below:$3$] (v3) at (2*\len, 0) {};
\node [label=below:$4$] (v4) at (3*\len, 0) {};
\node [label=below:$5$] (v5) at (4*\len, 0) {};
\node [label=below:$6$] (v6) at (5*\len, 0) {};
\draw [] (v1) to [out=45, in=135] (v3);
\draw [] (v1) to [out=45, in=135] (v4);
\draw [] (v1) to [out=45, in=135] (v5);
\draw [] (v1) to [out=45, in=135] (v6);
\draw [] (v2) to [out=45, in=135] (v4);
\draw [] (v2) to [out=45, in=135] (v6);
\draw [] (v3) to (v4);
\draw [] (v5) to (v6);
\end{tikzpicture}}
  \caption{
    \subref{fig:PI-graph} A simple-triangle graph $G_2$. 
    \subref{fig:PI-model} A model of $G_2$. 
    \subref{fig:co-PI-graph} The complement $\overline{G_2}$ of $G_2$. 
    \subref{fig:co-PI-ordering} An ordering of the vertices of $\overline{G_2}$. 
    The vertices are labeled based on the points on $L_1$. 
    As shown in Example~\ref{ex:PI}, 
    the word $w = 464365235121$ is a 12-representant of $\overline{G_2}$. 
  }
  \label{fig:PI}
\end{figure}

Simple-triangle graphs were introduced in~\cite{CK87-CN} 
as a generalization of both interval graphs and permutation graphs 
and have been studied under \emph{PI graphs}~\cite{BLS99,Spinrad03}, 
where \emph{PI} stands for \emph{Point-Interval}. 
The recognition of simple-triangle graphs has been a longstanding open problem~\cite[Open Problem 13.3]{Spinrad03}, 
and some polynomial-time recognition algorithms have been presented 
recently~\cite{Mertzios15-SIAMDM,Takaoka20-DAM,Takaoka20a-DAM}. 
The class of simple-triangle graphs is 
known to be a proper subclass of trapezoid graphs~\cite{CK87-CN}. 
It is also known that a simple-triangle graph is 
a cocomparability graph and alternately orientable~\cite{Takaoka18-DM}.

\subsection{Vertex ordering characterizations}
Recall that the necessary condition for 12-representability 
(i.e., Theorem~\ref{thm:good labeling}) is stated in terms of graph labelings. 
A labeling of a graph $G$ can be viewed as ordering the vertices of $G$ 
such that $x \prec y$ in the ordering if the label of $x$ is smaller than that of $y$. 
Then, the graphs $I_3$, $J_4$, and $Q_4$ in Figure~\ref{fig:I3 J4 Q4} 
correspond to ordered graphs in~\ref{fig:VOC PI}\subref{fig:cp},~\ref{fig:VOC PI}\subref{fig:p1}, 
and~\ref{fig:VOC PI}\subref{fig:p2}, respectively. 
We will use characterizations of interval containment bigraphs and simple-triangle graphs 
defined in terms of forbidden ordered induced subgraphs. 
We will refer to such an ordered graph as a \emph{pattern}. 
\par
An example of forbidden pattern characterization is as follows. 
A graph $G$ is a \emph{comparability graph} if each edge can be oriented 
so that if $x \to y$ and $y \to z$ then $x \to z$ for any $x, y, z \in V(G)$. 
It is known that a graph $G$ is a comparability graph if and only if 
there is a vertex ordering $\sigma$ of $G$ such that 
for any $x, y, z \in V(G)$ with $x \prec y \prec z$ in $\sigma$, 
if $xy \in E(G)$ and $yz \in E(G)$ then $xz \in E(G)$. 
In other words, a graph is a comparability graph if and only if 
it has a vertex ordering which does not contain 
the pattern in Figure~\ref{fig:VOC PI}\subref{fig:cp} as an induced pattern. 
Other examples can be found in~\cite[Section 7.4]{BLS99} and \cite{FH21-SIDMA}. 

\begin{theorem}[\cite{Takaoka18-DM}]\label{thm:VOC PI}
A graph $G$ is a simple-triangle graph if and only if 
the complement $\overline{G}$ of $G$ has a vertex ordering 
which does not contain any pattern in Figure~\ref{fig:VOC PI} 
as an induced pattern. 
Moreover, for any such ordering $\sigma$, there is a model of $G$ such that 
$\sigma$ coincides with the ordering of the points on $L_1$. 
(Recall that the triangle in the model is defined by a point on $L_1$ and an interval on $L_2$.) 
Such a model of $G$ can be obtained in $O(n^2)$ time if $\sigma$ is given. 
\end{theorem}

\begin{figure}[ht]
  \centering
  \subcaptionbox{\label{fig:cp}}{\begin{tikzpicture}
\useasboundingbox (-1.6, -0.1) rectangle (1.6, 0.9);
\tikzstyle{every node}=[draw,circle,fill=white,minimum size=5pt,
                        inner sep=0pt]
\node [label=below:$$] (x) at (-1, 0) {};
\node [label=below:$$] (y) at ( 0, 0) {};
\node [label=below:$$] (z) at ( 1, 0) {};
\draw [] (x) -- (y);
\draw [] (y) -- (z);
\end{tikzpicture}}
  \subcaptionbox{\label{fig:p1}}{\begin{tikzpicture}
\useasboundingbox (-2.1, -0.1) rectangle (2.1, 0.9);
\tikzstyle{every node}=[draw,circle,fill=white,minimum size=5pt,
                        inner sep=0pt]
\node [label=below:$$] (x) at (-1.5, 0) {};
\node [label=below:$$] (y) at (-0.5, 0) {};
\node [label=below:$$] (z) at ( 0.5, 0) {};
\node [label=below:$$] (w) at ( 1.5, 0) {};
\draw [] (x) to [out=60, in=120] (z);
\draw [] (y) to [out=60, in=120] (w);
\end{tikzpicture}}
  \subcaptionbox{\label{fig:p2}}{\begin{tikzpicture}
\useasboundingbox (-2.1, -0.1) rectangle (2.1, 0.9);
\tikzstyle{every node}=[draw,circle,fill=white,minimum size=5pt,
                        inner sep=0pt]
\node [label=below:$$] (x) at (-1.5, 0) {};
\node [label=below:$$] (y) at (-0.5, 0) {};
\node [label=below:$$] (z) at ( 0.5, 0) {};
\node [label=below:$$] (w) at ( 1.5, 0) {};
\draw [] (x) to [out=60, in=120] (w);
\draw [] (y) to (z);
\end{tikzpicture}}
  \caption{Forbidden patterns of complements of simple-triangle graphs. } 
  \label{fig:VOC PI}
\end{figure}
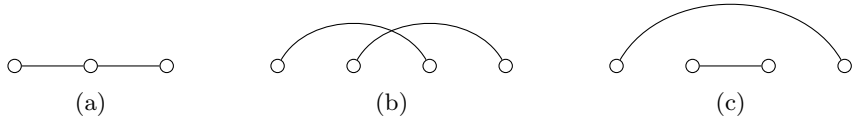

Figures~\ref{fig:PI}\subref{fig:co-PI-graph} and ~\ref{fig:PI}\subref{fig:co-PI-ordering} 
show the complement $\overline{G_2}$ of $G_2$ and its vertex ordering, 
which does not contain any pattern in Figure~\ref{fig:VOC PI}. 
\par
Notice that an $\{I_3, J_4, Q_4\}$-free labeling of a graph can be viewed as 
a vertex ordering which does not contain any pattern in Figure~\ref{fig:VOC PI}. 
Thus, we have from Theorems~\ref{thm:good labeling} and~\ref{thm:VOC PI} 
that any 12-representable graph is the complement of a simple-triangle graph. 
\par
For interval containment bigraphs, the following characterization is known. 
\begin{theorem}[\cite{Huang18-DAM}]\label{thm:VOC ICB}
A bipartite graph $G$ with bipartition $(X, Y)$ is 
an interval containment bigraph if and only if 
$G$ has a vertex ordering which does not contain any pattern in Figure~\ref{fig:VOC ICB} 
as an induced pattern. 
\end{theorem}

\begin{figure}[ht]
  \centering
  \subcaptionbox{}{\begin{tikzpicture}
\useasboundingbox (-2.1, -0.1) rectangle (2.1, 0.9);
\tikzstyle{every node}=[draw,circle,fill=white,minimum size=5pt,
                        inner sep=0pt]
\node [label=below:$$] (x) at (-1.5, 0) {};
\node [label=below:$$] (y) at (-0.5, 0) {};
\node [label=below:$$,fill=black] (z) at ( 0.5, 0) {};
\node [label=below:$$,fill=black] (w) at ( 1.5, 0) {};
\draw [] (x) to [out=60, in=120] (z);
\draw [] (y) to [out=60, in=120] (w);
\end{tikzpicture}}
  \subcaptionbox{}{\begin{tikzpicture}
\useasboundingbox (-2.1, -0.1) rectangle (2.1, 0.9);
\tikzstyle{every node}=[draw,circle,fill=white,minimum size=5pt,
                        inner sep=0pt]
\node [label=below:$$] (x) at (-1.5, 0) {};
\node [label=below:$$,fill=black] (y) at (-0.5, 0) {};
\node [label=below:$$] (z) at ( 0.5, 0) {};
\node [label=below:$$,fill=black] (w) at ( 1.5, 0) {};
\draw [] (x) to [out=60, in=120] (w);
\draw [] (y) to (z);
\end{tikzpicture}}
  \subcaptionbox{}{\begin{tikzpicture}
\useasboundingbox (-2.1, -0.1) rectangle (2.1, 0.9);
\tikzstyle{every node}=[draw,circle,fill=white,minimum size=5pt,
                        inner sep=0pt]
\node [label=below:$$] (x) at (-1.5, 0) {};
\node [label=below:$$] (y) at (-0.5, 0) {};
\node [label=below:$$,fill=black] (z) at ( 0.5, 0) {};
\node [label=below:$$,fill=black] (w) at ( 1.5, 0) {};
\draw [] (x) to [out=60, in=120] (w);
\draw [] (y) to (z);
\end{tikzpicture}}
  \caption{
    Forbidden patterns of interval containment bigraphs. 
    White and black vertices are in $X$ and $Y$, respectively, or the other way around. 
    }
  \label{fig:VOC ICB}
\end{figure}
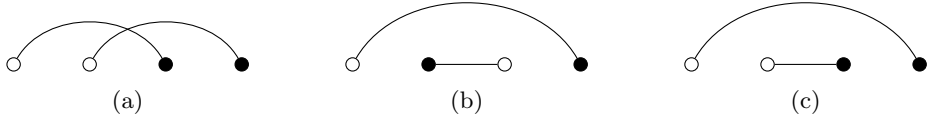

Figure~\ref{fig:ICB}\subref{fig:ICB-ordering} shows a vertex ordering of $G_1$, 
which does not contain any pattern in Figure~\ref{fig:VOC ICB}. 
\par
Notice that a $\{J_4, Q_4\}$-free labeling of a bipartite graph can be viewed as 
a vertex ordering which does not contain any pattern in Figure~\ref{fig:VOC ICB}. 
Thus, we have from Theorems~\ref{thm:good labeling} and~\ref{thm:VOC ICB} 
that any 12-representable bipartite graph is an interval containment bigraph.

\section{Interval containment bigraphs}\label{sec:ICB}
This section shows the equivalence of 
12-representable bipartite graphs and interval containment bigraphs 
and its consequences. 

\begin{theorem}\label{thm:ICB}
Let $G$ be a bipartite graph. 
The following statements are equivalent: 
\begin{enumerate}[label=\textup{(\roman*)}]
\item $G$ is $12$-representable; 
\item there is a $\{J_4, Q_4\}$-free labeling of $G$; 
\item $G$ is an interval containment bigraph. 
\end{enumerate}
\end{theorem}
\begin{proof}
The implications (i) $\implies$ (ii) and (ii) $\implies$ (iii) 
follow from Theorems~\ref{thm:good labeling} and~\ref{thm:VOC ICB}, respectively. 
To prove (iii) $\implies$ (i), we construct a labeling and 
a 12-representant of an interval containment bigraph. 
See Example~\ref{ex:ICB} for an instance of construction. 
\par
Let $G$ be an interval containment bigraph with bipartition $(X, Y)$ 
such that there is an interval $I_v$ for each $v \in V(G)$ and 
$xy \in E(G) \iff I_x \supseteq I_y$ for any $x \in X$ and $y \in Y$. 
As stated in~\cite{Huang06-JCTSB}, 
it is possible to choose intervals so that all endpoints are distinct. 
Thus, without loss of generality, we can assume that all endpoints are distinct. 
Let $\ell_v$ and $r_v$ denote the left and right endpoint of the interval $I_v$, 
respectively. 
We assign a label $i$ to a vertex $v \in V(G)$ 
if $\ell_v$ is the $i$th point among all left endpoints from left to right. 
\par
Let $\pi_r$ be a permutation of $[n]$ such that 
the $i$th letter of $\pi_r$ is the label of a vertex $v$ 
if $r_v$ is the $i$th point among all right endpoints from left to right. 
Let $\pi_x$ and $\pi_y$ be arbitrary permutations of 
the labels of vertices of $X$ and $Y$, respectively. 
\par
We claim that $w = \pi_y\pi_r\pi_x$ is a 12-representant of $G$. 
Let $u$ and $v$ be two vertices of $G$ with labels $i$ and $j$, respectively. 
Without loss of generality, we assume $i < j$, that is, $\ell_u < \ell_v$. 
If $u, v \in X$ then $w_{\{i, j\}}$ has a 12-match 
since both $\pi_r$ and $\pi_x$ contain $i$ and $j$. 
Similarly, if $u, v \in Y$ then $w_{\{i, j\}}$ has a 12-match
since both $\pi_y$ and $\pi_r$ contain $i$ and $j$. 
Suppose $u \in X$ and $v \in Y$. 
If $r_u > r_v$ then $w_{\{i, j\}} = jjii$, and 
if $r_u < r_v$ then $w_{\{i, j\}} = jiji$. 
Thus, $w_{\{i, j\}}$ has no 12-match if and only if $I_u$ contains $I_v$. 
If $u \in Y$ and $v \in X$, then $w_{\{i, j\}}$ has a 12-match 
since $\pi_y$ is to the left of $\pi_x$, 
which is consistent with the fact that $I_v$ does not contain $I_u$. 
\end{proof}

\begin{exm}\label{ex:ICB}
The graph $G_1$ in Figure~\ref{fig:ICB}\subref{fig:ICB-graph} is an interval containment bigraph. 
The vertices are labeled based on the left endpoints of the intervals in Figure~\ref{fig:ICB}\subref{fig:ICB-model}. 
By reading the labels of the right endpoints from left to right, 
we obtain the permutation $\pi_r = 53284761$. 
Let $\pi_x = 1246$ and $\pi_y = 3578$. 
It is straightforward to check that the word 
$w = 3578.53284761.1246$ is a 12-representant of $G_1$ 
(the dots are not part of the word, they are only included as delimiters of the word parts as constructed in the proof of Theorem~\ref{thm:ICB}). 
\end{exm}

Recall that the class of interval containment bigraphs coincides with 
the class of bipartite graphs whose complements are circular-arc graphs~\cite{FHH99-Combinatorica} and 
the class of two-directional orthogonal ray graphs~\cite{STU10-DAM}. 
As stated in~\cite{FHH99-Combinatorica,STU10-DAM}, 
Trotter and Moore~\cite{TM76-DM} provide 
the list of minimal forbidden induced subgraphs for 
bipartite graphs whose complements are circular-arc graphs. 
Therefore, Theorem~\ref{thm:ICB} provides 
a forbidden induced subgraph characterization 
of 12-representable bipartite graphs. 
See~\cite{STU10-DAM} for figures of the forbidden subgraphs. 
\par
From the list of forbidden induced subgraphs for 12-representable bipartite graphs, 
we also have a characterization of 12-representable grid graphs. 
\begin{cor}\label{cor:Forbidden subgraphs}
A grid graph is $12$-representable if it contains 
no cycle of length $2n$ for $n \geq 4$ and 
no graph in Figure~\ref{fig:Forbidden subgraphs} as an induced subgraph. 
\end{cor}
\begin{proof}
It is easy to verify that the other graphs in the list of forbidden induced subgraphs for 12-representable bipartite graphs (see~\cite{STU10-DAM} for figures) are not 
induced subgraphs of a rectangular grid graph. 
\end{proof}

\begin{figure}[ht]
  \centering
  \subcaptionbox{\label{fig:T3}}{\begin{tikzpicture}
\def\len{0.4}
\useasboundingbox (-1.2, -1.0) rectangle (1.2, 1.4);
\tikzstyle{every node}=[draw,circle,fill=white,minimum size=3pt,inner sep=0pt]
\node [] (c)  at (0, 0) {};
\node [] (a1) at (90:\len*1) {};
\node [] (a2) at (90:\len*2) {};
\node [] (a3) at (90:\len*3) {};
\node [] (b1) at (210:\len*1) {};
\node [] (b2) at (210:\len*2) {};
\node [] (b3) at (210:\len*3) {};
\node [] (c1) at (330:\len*1) {};
\node [] (c2) at (330:\len*2) {};
\node [] (c3) at (330:\len*3) {};
\draw [] (c) -- (a1) -- (a2) -- (a3);
\draw [] (c) -- (b1) -- (b2) -- (b3);
\draw [] (c) -- (c1) -- (c2) -- (c3);
\end{tikzpicture}}
  \subcaptionbox{\label{fig:G2}}{\begin{tikzpicture}
\def\len{0.5}
\useasboundingbox (-\len-0.1, -0.2) rectangle (\len+0.1, 4*\len+0.2);
\tikzstyle{every node}=[draw,circle,fill=white,minimum size=3pt,inner sep=0pt]
\node [] (v1) at (0, 1*\len) {};
\node [] (v2) at (0, 2*\len) {};
\node [] (v3) at (0, 3*\len) {};
\node [] (v4) at (0, 4*\len) {};
\node [] (v5) at (\len, 0) {};
\node [] (v6) at (\len, 1*\len) {};
\node [] (v7) at (\len, 2*\len) {};
\node [] (v8) at (-\len, 0) {};
\node [] (v9) at (-\len, 1*\len) {};
\node [] (v10) at (-\len, 2*\len) {};
\draw [] 
	(v1) -- (v2) -- (v3) -- (v4)
	(v5) -- (v6) -- (v7)
	(v8) -- (v9) -- (v10)
	(v9) -- (v1) -- (v6)
	(v10) -- (v2) -- (v7)
;
\end{tikzpicture}}
  \subcaptionbox{\label{fig:G3}}{\begin{tikzpicture}
\def\len{0.5}
\useasboundingbox (-0.1, -0.2) rectangle (\len+0.1, 4*\len+0.2);
\tikzstyle{every node}=[draw,circle,fill=white,minimum size=3pt,inner sep=0pt]
\node [] (v1) at (0, 0) {};
\node [] (v2) at (0, 1*\len) {};
\node [] (v3) at (0, 2*\len) {};
\node [] (v4) at (0, 3*\len) {};
\node [] (v5) at (0, 4*\len) {};
\node [] (u1) at (\len, 0) {};
\node [] (u2) at (\len, 1*\len) {};
\node [] (u3) at (\len, 2*\len) {};
\node [] (u4) at (\len, 3*\len) {};
\node [] (u5) at (\len, 4*\len) {};
\draw [] 
	(v1) -- (v2) -- (v3) -- (v4) -- (v5)
	(u1) -- (u2) -- (u3) -- (u4) -- (u5)
	(v3) -- (u3)
	(v4) -- (u4)
	(v5) -- (u5)
;
\end{tikzpicture}}
  \subcaptionbox{}{\begin{tikzpicture}
\def\len{0.5}
\useasboundingbox (-1.2, -1.2) rectangle (1.2, 1.2);
\tikzstyle{every node}=[draw,circle,fill=white,minimum size=3pt,inner sep=0pt]
\node [] (a) at (0, 0) {};
\node [] (b) at (0, 1*\len) {};
\node [] (c) at (0, 2*\len) {};
\node [] (d) at ($(225:\len*0) + (315:\len*1)$) {};
\node [] (e) at ($(225:\len*0) + (315:\len*2)$) {};
\node [] (f) at ($(225:\len*0) + (315:\len*3)$) {};
\node [] (g) at ($(225:\len*1) + (315:\len*0)$) {};
\node [] (h) at ($(225:\len*2) + (315:\len*0)$) {};
\node [] (i) at ($(225:\len*3) + (315:\len*0)$) {};
\node [] (j) at ($(225:\len*1) + (315:\len*1)$) {};
\draw [] 
	(a) -- (b) -- (c)
	(a) -- (d) -- (e) -- (f)
	(a) -- (g) -- (h) -- (i)
	(d) -- (j) -- (g)
;
\end{tikzpicture}}
  \subcaptionbox{}{\begin{tikzpicture}
\def\len{0.5}
\useasboundingbox (-3*\len-0.1, -1.5*\len-0.2) rectangle (3*\len+0.1, 2.5*\len+0.2);
\tikzstyle{every node}=[draw,circle,fill=white,minimum size=3pt,inner sep=0pt]
\node [] (v1) at (0, 2*\len) {};
\node [] (v2) at (0, 1*\len) {};
\node [] (v3) at (-3*\len, 0) {};
\node [] (v4) at (-2*\len, 0) {};
\node [] (v5) at (-1*\len, 0) {};
\node [] (v6) at ( 0*\len, 0) {};
\node [] (v7) at ( 1*\len, 0) {};
\node [] (v8) at ( 2*\len, 0) {};
\node [] (v9) at ( 3*\len, 0) {};
\node [] (v10) at (-1*\len, -1*\len) {};
\node [] (v11) at ( 0*\len, -1*\len) {};
\node [] (v12) at ( 1*\len, -1*\len) {};
\draw [] 
	(v1) -- (v2) -- (v6)
	(v3) -- (v4) -- (v5) -- (v6) -- (v7) -- (v8) -- (v9)
	(v10) -- (v11) -- (v12)
	(v5) -- (v10)
	(v6) -- (v11)
	(v7) -- (v12)
;
\end{tikzpicture}}
  \subcaptionbox{\label{fig:X}}{\begin{tikzpicture}
\def\len{0.5}
\useasboundingbox (-0.8, -0.7) rectangle (0.8, 1.7);
\tikzstyle{every node}=[draw,circle,fill=white,minimum size=3pt,inner sep=0pt]
\node [] (b) at ($(45:\len*1) + (135:\len*1)$) {};
\node [] (a) at ($(b) + (0, \len)$) {};
\node [] (c) at ($(45:\len*1) + (135:\len*0)$) {};
\node [] (d) at ($(45:\len*1) + (135:\len*-1)$) {};
\node [] (e) at ($(45:\len*0) + (135:\len*-1)$) {};
\node [] (f) at (0, 0) {};
\node [] (g) at ($(45:\len*0) + (135:\len*1)$) {};
\node [] (h) at ($(45:\len*-1) + (135:\len*1)$) {};
\node [] (i) at ($(45:\len*-1) + (135:\len*0)$) {};
\draw [] 
	(a) -- (b)
	(b) -- (c) -- (d)
	(e) -- (f) -- (g)
	(h) -- (i)
	(b) -- (g) -- (h)
	(c) -- (f) -- (i)
	(d) -- (e)
;
\end{tikzpicture}}
  \caption{Forbidden induced subgraphs of 12-representable grid graphs. }
  \label{fig:Forbidden subgraphs}
\end{figure}
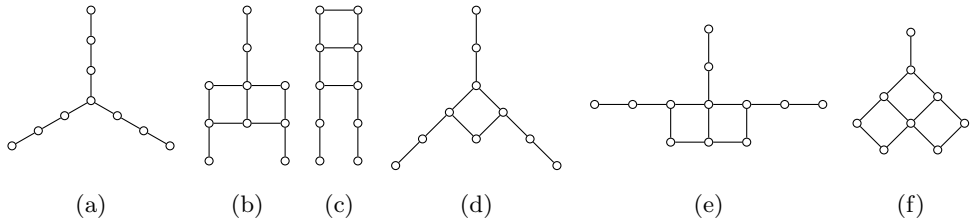

\begin{rem}\label{rem:Forbidden subgraphs}
Chen and Kitaev~\cite{CK22-DMGT} presented certain non-$12$-representable graphs 
and conjectured that these graphs would give us 
a forbidden induced subgraph characterization of 12-representable line grid graphs, 
see~\cite[Conjecture 3.6 and Figure 3.26]{CK22-DMGT}. 
Corollary~\ref{cor:Forbidden subgraphs} indicates that 
the graphs in~\cite[Conjecture 3.6]{CK22-DMGT} are not sufficient 
to characterize 12-representable line grid graphs. 
For example, the graph in Figure~\ref{fig:Forbidden subgraphs}\subref{fig:G2} 
is a proper induced subgraph of $G_i$, $i \in \{3, 4, 5\}$ in~\cite[Figure 3.26]{CK22-DMGT} and 
the graph in Figure~\ref{fig:Forbidden subgraphs}\subref{fig:G3} 
is a proper induced subgraph of $G_6$ in~\cite[Figure 3.26]{CK22-DMGT}. 
\end{rem}

Interval containment bigraphs can be recognized in $O(n^2)$ time~\cite{STU10-DAM} 
because their complements 
(i.e., circular-arc graphs that can be partitioned into two cliques) 
can be recognized in $O(n^2)$ time~\cite{ES93-SODA},~\cite[Section 13.3]{Spinrad03}. 
Thus, Theorem~\ref{thm:ICB} yields the following. 
\begin{cor}
$12$-representable bipartite graphs can be recognized in $O(n^2)$ time. 
\end{cor}

A graph $G$ is a \emph{circular-arc graph} 
if there is a circular arc $A_v$ on a circle for each vertex $v \in V(G)$ 
such that for any $u, v \in V(G)$, 
$uv \in E(G)$ if and only if $A_u$ intersects $A_v$. 
The set $\{A_v \colon\ v \in V(G)\}$ is called a 
\emph{model} or \emph{representation} of $G$. 
If the given bipartite graph $G$ is the complement of a circular-arc graph, 
the recognition algorithm~\cite{ES93-SODA},~\cite[Section 13.3]{Spinrad03} 
provides a model of $\overline{G}$. 
The model can be easily transformed into a model of interval containment bigraphs~\cite{Huang06-JCTSB}. 
Thus, we have the following from Theorem~\ref{thm:ICB}. 
\begin{cor}
A $12$-representant of a bipartite graph can be obtained in $O(n^2)$ time 
if the graph is $12$-representable. 
\end{cor}

\section{Simple-triangle graphs}
This section shows the equivalence of 
12-representable graphs and complements of simple-triangle graphs 
and its consequences. 

\begin{theorem}\label{thm:PI}
Let $G$ be a graph. 
The following statements are equivalent: 
\begin{enumerate}[label=\textup{(\roman*)}]
\item $G$ is $12$-representable; 
\item there is an $\{I_3, J_4, Q_4\}$-free labeling of $G$; 
\item the complement $\overline{G}$ of $G$ is a simple-triangle graph. 
\end{enumerate}
\end{theorem}
\begin{proof}
The implications (i) $\implies$ (ii) and (ii) $\implies$ (iii) 
follow from Theorems~\ref{thm:good labeling} and~\ref{thm:VOC PI}, respectively. 
To prove (iii) $\implies$ (i), we construct a labeling and 
a 12-representant of the complement of a simple-triangle graph. 
See Example~\ref{ex:PI} for an instance of construction. 
\par
Recall that $L_1$ and $L_2$ are two horizontal lines 
in the plane with $L_1$ above $L_2$. 
Let $G$ be a simple-triangle graph such that 
there is a triangle $T_v$ for each $v \in V(G)$ and 
$uv \in E(G) \iff T_u \cap T_v \neq \emptyset$ for any $u, v \in V(G)$. 
Without loss of generality, we can assume that the endpoints of the triangles are distinct. 
Let $p_v$ and $I_v$ be the point on $L_1$ and the interval on $L_2$ of $T_v$, respectively. 
We assign a label $i$ to a vertex $v \in V(G)$ 
if $p_v$ is the $i$th point on $L_1$ from left to right. 
\par
We form a word $w$ using the endpoints of the intervals on $L_2$ so that 
the $i$th letter of $w$ is the label of a vertex $v$ if, 
among all endpoints of the intervals (i.e., both left and right endpoints) 
\emph{from right to left}, the $i$th endpoint is of $I_v$. 
We claim that $w$ is a 12-representant of the complement $\overline{G}$ of $G$. 
Let $u$ and $v$ be two vertices of $G$ with labels $i$ and $j$, respectively. 
Without loss of generality, we assume $i < j$, that is, $p_u < p_v$. 
It is easy to see that 
$w_{\{i, j\}} = jjii$ if and only if $I_u$ lies entirely to the left of $I_v$. 
Thus, $uv \in E(\overline{G})$ if and only if $w_{\{i, j\}}$ has no 12-match. 
\end{proof}

\begin{exm}\label{ex:PI}
The graph $G_2$ in Figure~\ref{fig:PI}\subref{fig:PI-graph} is a simple-triangle graph. 
The vertices are labeled based on the points on $L_1$ in Figure~\ref{fig:PI}\subref{fig:PI-model}. 
By reading the labels of the endpoints on $L_2$ \emph{from right to left}, 
we obtain the word $w = 464365235121$. 
It is straightforward to check that the word $w$ 
is a 12-representant of the complement $\overline{G_2}$ of $G_2$. 
\end{exm}

By Theorems~\ref{thm:VOC PI} and~\ref{thm:PI}, we have the following. 
\begin{cor}\label{cor:representant}
From an $\{I_3, J_4, Q_4\}$-free labeling of a $12$-representable graph $G$, 
a $12$-representant of $G$ can be obtained in $O(n^2)$ time 
without relabeling of $G$. 
\end{cor}
\begin{proof}
An $\{I_3, J_4, Q_4\}$-free labeling of a graph $G$ can be viewed as 
a vertex ordering $\sigma$, which does not contain any pattern in Figure~\ref{fig:VOC PI}. 
Thus, by Theorem~\ref{thm:VOC PI}, we can obtain a model of the complement $\overline{G}$ of $G$ 
in $O(n^2)$ time such that 
$\sigma$ coincides with the ordering of the points on $L_1$. 
A 12-representant of $G$ can be obtained from the model, 
as shown in the proof of Theorem~\ref{thm:PI}. 
\end{proof}

Theorem~\ref{thm:PI} also yields the following, 
since simple-triangle graphs can be recognized in $O(nm)$ time~\cite{Takaoka20a-DAM}
and the complement of a graph can be obtained in $O(n^2)$ time. 
\begin{cor}
$12$-representable graphs can be recognized in $O(n(\bar{m}+n))$ time, 
where $\bar{m}$ is the number of edges of the complement of the given graph. 
\end{cor}

The recognition algorithm~\cite{Takaoka20a-DAM} provides a vertex ordering 
which does not contain any pattern in Figure~\ref{fig:VOC PI}, and 
we have the following from Corollary~\ref{cor:representant}. 
\begin{cor}
A $12$-representant of a graph can be obtained in $O(n(\bar{m}+n))$ time 
if the graph is $12$-representable. 
\end{cor}

\section{Concluding remarks}
The 12-representants constructed in the proof of Theorems~\ref{thm:ICB} and \ref{thm:PI} 
are of length $2n$, but they are not necessarily optimal (shortest possible). 
Indeed, for example, as shown in~\cite[Theorem~2.18]{CK22-DMGT}, 
the graph $G_1$ in Figure~\ref{fig:ICB}\subref{fig:ICB-graph} can be 
12-represented by a word of length $n+1$ 
(the labeling used in~\cite[Theorem~2.18]{CK22-DMGT} is different 
from that shown in Figure~\ref{fig:ICB}\subref{fig:ICB-graph}). 
It is still an open question to improve the upper bound of 
the length of 12-representants of graphs. 
\par
Section~\ref{sec:ICB} gives a forbidden induced subgraph characterization 
for 12-representable bipartite graphs and grid graphs 
from the equivalence between 12-representable bipartite graphs and interval containment bigraphs. 
Although the characterization has been known for interval containment bigraphs, 
no such characterization is known for simple-triangle graphs~\cite{Takaoka20a-DAM}. 
Thus, it is still an open question to characterize the class of 12-representable graphs 
in terms of forbidden induced subgraphs. 
\par
In this paper, we obtained some results on 12-representable graphs 
from the known facts on interval containment bigraphs and simple-triangle graphs. 
Studying these graphs via 12-representability 
is a possible direction for further research.

\section*{Acknowledgments}
The author is grateful to the reviewers for their careful reading and helpful comments.

\bibliographystyle{abbrvurl}
\bibliography{ref}
\end{document}